\documentclass[submission,copyright,creativecommons]{eptcs}
\providecommand{\event}{NCMA 2025} 

\usepackage{graphicx}
\usepackage{amsmath,amssymb}
\usepackage{pifont}
\usepackage{color}
\providecommand{\thisvolume}[1]{this volume of EPTCS, Open Publishing Association}

\ifpdf
  \usepackage{underscore}         
  \usepackage[T1]{fontenc}        
\else
  \usepackage{breakurl}           
\fi
\newtheorem{definition}{Definition}
\newtheorem{theorem}{Theorem}
\newtheorem{lemma}{Lemma}
\newtheorem{example}{Example}

\newtheorem{corollary}{Corollary}

\newtheorem{proposition}{Proposition}

\newenvironment{proof}{\begin{trivlist}
                       \item[]{\bf Proof}
                       \hspace{0cm} }{\hfill {\large $\bullet$}
                       \end{trivlist}}

\newcommand{\be}{\begin{eqnarray}}
\newcommand{\ee}{\end{eqnarray}}

\newcommand{\ignore}[1]{}

\newcommand{\qed}{}

\title{A Myhill-Nerode Type Characterization of 2detLIN Languages}
\author{Benedek Nagy
\institute{Department of Mathematics,
Eastern Mediterranean University\\
99628 Famagusta, North Cyprus, Mersin-10, Turkey\\
Department of Computer Science, Institute of Mathematics and Informatics,\\ Eszterh\'azy K\'aroly Catholic University, Eger, Hungary}
\email{nbenedek.inf@gmail.com}
}

\def\titlerunning{A Myhill-Nerode Type Characterization of 2detLIN Languages}
\def\authorrunning{Benedek Nagy}
\begin{document}

\maketitle

\begin{abstract}
Linear automata are automata with two reading heads starting from the two extremes of the input,
 are equivalent to
$5'\to 3'$
Watson-Crick (WK) finite automata.
The heads read the input in opposite directions and the computation finishes when the heads meet.
 These automata %
  accept the class LIN of linear languages.
 The deterministic counterpart of these models, on the one hand, is less expressive, as only a proper subset of LIN, the class 2detLIN is accepted; and on the other hand, they are also equivalent %
  in the sense of the class of the accepted languages.
  Now, based on these automata models, we characterize the class of 2detLIN languages with a Myhill-Nerode type of equivalence classes. However, as these automata may do the computation of both the prefix and the suffix of the input, we use prefix-suffix pairs in our classes. Additionally, it is proven that finitely many classes in the characterization match with the 2detLIN languages, but we have some constraints on the used prefix-suffix pairs, i.e., the characterization should have the property to be complete and it must not have any crossing pairs.
\end{abstract}


\section{Introduction}
\label{s:intro}

In formal language theory, the class of regular languages plays a crucial role, similar as finite automata in automata theory.
They are widely applied and there are several theoretical studies known about them. One important fact is the characterization of regular languages by the Myhill-Nerode theorem \cite{Myhill,Nerode}.
In a nutshell, every regular language induces finitely many equivalence classes of words considering them as possible prefixes of the words of the language.
This ``if and only if'' characterization, in fact, gives also the minimal completely defined deterministic finite automaton for each regular language and thus, it has very important practical consequences.
The number of states of such minimal automaton  is the same as the number of equivalence classes above, for each language.
This measure is the most known and most used measure for descriptional complexity of regular languages. There are other known measures, e.g., transition complexity \cite{ShengDCFS2010}, nondeterministic state and transition complexities \cite{KaiDLT2007}, union-complexity \cite{DCFS,InfComp}, just to mention a few.

In this paper, we consider a proper superclass of the class of regular languages %
based on a kind of deterministic 2-head automata. This model starts  the computation %
by having its two heads at the two extremes of the input: the first head may read the first and the second head may read the last letter of the input. The computation goes step by step %
till the heads meet (at some position of the input). If the automaton is in a final state
at that time, %
then the computation is accepting and the input is in the accepted language. The class of the nondeterministic variant of these automata  accepts the class of linear languages, another well-known class of formal languages. It is properly between the regular and context-free classes.
Here, usually, we refer to this model of automata as linear automata (based on \cite{Rouss,RAIROlinTrace}). However, very similar models were defined also under various names, e.g., 2-head automata \cite{triangle} %
or biautomata \cite{biAUTncma}. \\
We also recall the concept of Watson-Crick finite automata
which belongs to a special field of DNA
computing. %
From the end of the last century, DNA computing has emerged as a relatively new computational paradigm \cite{DNA-book}.
Watson-Crick automata (abbreviated as WK automata) have been introduced in \cite{Freund},
for details and early results see also \cite{DNA-book}. %
A WK automaton works on a double-stranded tape called Watson-Crick tape (i.e., on a DNA molecule), whose strands are scanned separately by read only heads. The symbols in the corresponding cells of the double-stranded tapes are related by the Watson-Crick complementarity relation (a symmetric and bijective relation in the nature with pairs Adenine-Thymine and Cytosine-Guanine). %
The two strands of a DNA molecule have opposite $5'\rightarrow 3'$  orientations. The %
$5'\rightarrow 3'$  WK automata are more realistic in the sense that both heads use the same biochemical direction (that is, actually, opposite physical directions) \cite{LeupoldN09,Leupold,DNA2008}. A WK automaton is sensing if it knows whether the heads are at the same position.
The sensing $5'\to 3'$ WK finite automata work essentially in the same way as linear automata, but they may read strings in a transition.
Their 1-limited variant, in which exactly one letter is read in each transition,
has the same power, i.e., they accept the same family of languages as the original model (\cite{CiE2009,NaCoNP22,AFL2017}). There are numerous variants of these automata where some extensions or restrictions are applied including stateless \cite{Annales}, state- and quasi-deterministic and reversible variants \cite{state-det,quasi-det,NCMA25-rev2head}, jumping $5'\to 3'$ WK automata \cite{KocmanKMN22}, as well as,
$5'\to 3'$ WK multi-counter and pushdown automata \cite{EgeciogluHN10,EgeciogluHN11,HegedusNE12,2headpushdown} and $5'\to 3'$ WK automata accepting necklaces \cite{NCMA24neck},
just to mention a few. %

 We are interested in a proper subclass of the linear languages, namely 2detLIN, the class that is accepted by the deterministic variant of the linear automata (and of the sensing $5' \to 3'$ WK automata), as they are described in details in, e.g., \cite{Nagy2013,ActInf2detLIN}.
 This class is still a proper superset of the class of regular languages.
 Here, we give a characterization of 2detLIN that is somewhat similar to Myhill-Nerode characterization of the regular languages.
 It is done by using prefix-suffix pairs. We show some important properties of the pairs that can be used in the characterization.
 Although there are significant differences between the original Myhill-Nerode characterization result and our result, we believe that our results could lead to a kind of similar descriptional complexity measure to a larger class of languages than the original results which can be used for the class of regular languages.

 Because of the page limit some of the proofs are omitted. %

\section{Definitions and Preliminaries}
\label{s:pre}

We assume that the reader is familiar with the basic concepts of formal languages and automata, otherwise she or he is referred, e.g., to \cite{HopUl,Handb} for the concepts not explained in detail here. We denote the empty word by $\lambda$.
The set of nonnegative integers is denoted by $\mathbb{N}$.

There are various classes in the Chomsky hierarchy. We briefly recall here the classes of regular and linear languages.
A \emph{generative grammar} is a four tuple $(N,T,S,P)$ with two disjoint, finite, nonempty alphabets $N$ and $T$, where the former is called nonterminal alphabet, the latter is called terminal alphabet. The symbol $S\in N$ is the start (a.k.a. sentence) symbol and $P$ is the finite set of productions (a.k.a. rewriting rules). Each production is of the form $u\to v$ where $u$ must contain at least one nonterminal symbol.
A generative grammar $(N,T,S,P)$ is \emph{regular} (in some places they are also called right-linear) if each production of the grammar is in one of the following forms: $A\to w$ (with $A\in N, w\in T^*$) and $A \to wB$ (with $A,B\in N, w\in T^*$).
Further, a generative grammar   is \emph{linear} if each of its productions is in one of the following forms: $A\to w$ (with $A\in N, w\in T^*$) and $A \to u B v$ (with $A,B\in N, u,v\in T^*$).
Obviously, every regular grammar is also linear at the same time. These classes of grammars generate the classes of regular and linear languages, respectively.
We recall here some special linear grammars: if in a linear grammar for each production with a nonterminal on the right side $A\to uBv$, $|u| =n$, $|v|=m$ holds, then the grammar is called \emph{$k$-rated linear} with the value $k=\frac{m}{n}, \ (m,n\in \mathbb{N}, n\ne 0)$ \cite{Amar2,linPUMP}. These grammars generate \emph{$k$-rated linear languages}. The union of the sets of $k$-rated linear languages for any nonnegative rational value of $k$ is called the family of \emph{fix-rated linear languages}. Observe that, in fact, the $0$-rated linear grammars and languages are the regular grammars and languages. The $1$-rated linear grammars and languages are usually referred to as \emph{even-linear grammars and languages} (\cite{Amar,EVENlin}).

The classes of regular and linear languages can be accepted by the class of traditional finite automata and a class of 2-head automata, respectively.
Let us discuss, first, the case of regular languages.
Now, let us recall the concept of finite automata. A five tuple $A=(Q,T,q_0,\delta,F)$ is a \emph{finite automaton} with the finite nonempty set of states $Q$, with a finite nonempty input alphabet $T$, an initial state $q_0\in Q$, a set of final (a.k.a. accepting) states $F$ and a transition function $\delta$. The latter is defined, in general, as $\delta: Q\times (T\cup \{\lambda\}) \to 2^Q$. In this general case, the model is known as nondeterministic finite automata. There is a more restricted version of the finite automata with $\delta: Q\times T \to Q$ called \emph{deterministic finite automata}. Automata are used to accept formal languages.
 It is well-known that the classes of both the nondeterministic and deterministic finite automata  recognize exactly the class of regular languages.
 There are some 2-head extensions of these traditional models, that play a central role for us.
A five tuple $A'=(Q,T,q_0,\delta,F)$ is a 2-head finite automaton (a.k.a. \emph{linear automaton}, \cite{Rouss,RAIROlinTrace}) where $Q,T,q_0$ and $F$ have the same roles as in traditional finite automata, but $\delta$ is defined in a different way:
 $\delta: ( Q \times T \times \{\lambda \} \cup Q\times \{\lambda\} \times T ) \to 2^Q$.

Further, a configuration of a linear %
 automaton is a pair $(q,w)$ where $q$ is the current state of the automaton and $w$ is the part of the input word which has not been processed (read) yet. For $w'\in T^*$, $xy \in T$, $q,q'\in Q$, we write a computation step %
  between two configurations as:
$(q,xw'y)\Rightarrow(q',w' )$ if and only if $q'\in \delta(q,x,y)$. Notice that in such a computation step either $x\in T$ and $y=\lambda$ or $y\in T$ and $x=\lambda$, i.e., exactly one of the heads is reading an input letter.
We denote the reflexive and transitive closure of the relation $\Rightarrow$ (one step of a computation) %
 by $\Rightarrow^*$, and refer to it as the computation relation. Therefore, for a given $w\in T^*$, an accepting computation is a sequence of computation steps %
  of the form $(q_0,w) \Rightarrow^* (q,\lambda)$, starting from the initial state and ending in a state $q\in F$ with no input left. Finally,
the language accepted by a linear automaton $M$ is:
$$L(M)=\{ w\in T^*  | (q_0,w) \Rightarrow^* (q,\lambda),  q\in F\}.$$

Note that here we have defined a kind of restricted (1-limited) variant, where exactly one input letter is read in each transition. (In   the
more general variant %
 both heads may read a letter %
in a transition, however, the described model is equivalent to this %
 more general model %
 in the sense of the class of the accepted languages.)

 It is known that the class of  linear automata accepts the class of linear languages. Now, we are interested in the deterministic variant of them. As usual, we say that an automaton is deterministic if at any possible configuration %
  there is at most one way to continue the computation. The deterministic counterpart of linear automata can accept only a special subclass of the class of linear languages, the \emph{class 2detLIN}. It is known that this class is a superclass of the class of regular languages containing various interesting linear languages, including all fix-rated linear languages. On the other hand, 2detLIN is incomparable to the class detLIN, the class accepted by deterministic one-turn pushdown automaton (with the deterministic counterpart of another well-known automata model accepting the class of linear languages).

 We also recall that a closely related model, the sensing $5' \to 3'$ Watson-Crick finite automata work in a very similar manner.
 There is a very important difference between the 2-head automata model we have defined and the Watson-Crick automata, namely that the latter models are able to read strings in a transition. Nevertheless, in \cite{AFL2017} and in \cite{ActInf2detLIN} it is proven that this feature does not help for the model to accept  larger classes of languages than the classes LIN and 2detLIN, respectively. %
 Therefore, we may use the definition we gave above to define the class we are interested in.

Further, we may assume that all states of the automaton $A$ is reachable, i.e., for each state $q$, there is an input word $w_q$ such that the computation of $w_q$ ends in state $q$: $(q_0,w_q) \Rightarrow^* (q, \lambda)$. (Those states that are not reachable do not have any effect on the computations of the automaton, and thus, they can simply be removed form the set of states without changing the accepted language.) This assumption could be important when some properties of the automaton are analyzed, e.g., in the proposition below.

Formally, we can write that a linear automaton is deterministic, if and only if for each pair of $
w\in T^*$ and $
q \in Q$ there exists at most one $w' \in T^*$ and $q' \in Q$ such that $(q,w) \Rightarrow (q',w')$. This property is defined as a constraint on all possible computations of the automaton, however,
 it gives restriction for the used automaton itself. Thus,
deterministic linear automata can be characterized as follows. %

\begin{proposition}\label{linDET1head}
A linear automaton is deterministic if and only if for each $q$ of its states, either
\begin{itemize}
\item
  $\delta(q,a,\lambda) = \emptyset$ for all $a\in T$; 
 and  $| \delta(q,\lambda,a) | \leq 1$  for each $a\in T$; \\
or
 \item
 $\delta(q,\lambda,a) = \emptyset$ for all $a\in T$; 
and  $|\delta(q,a,\lambda) | \leq 1$  for each $a\in T$.
\end{itemize}
\end{proposition}

We refer to the transitions $\delta(q,a,b) = \emptyset$ ($ab\in T$, i.e., one of $a$ and $b$ is a letter, the other is $\lambda$), as transitions which are not defined in the automaton. Thus, we may interpret the previous statement as follows. In a deterministic linear automaton at each state, we may have transitions defined only for at most one of the heads.

In automata theory, there are usually two main variants of the used deterministic finite automata.
If a finite automaton is incomplete (we say this, when its transition function is only a partial function), it may happen that the automaton is unable to read (and thus to accept) some of the possible input words, in these cases, the automaton gets stuck and the computation halts without accepting. In contrast, in the case of a completely defined finite automaton, the automaton can read any input and can do the computation such that the whole input has been processed.
 Somewhat similarly, we may also define this variant of linear automata.
The main difference in the work of the  ``incomplete'' and ``completely defined'' (shortly, complete) linear automata is the same as at finite automata, however, based on Proposition \ref{linDET1head}, we may characterize the latter ones as follows.

\begin{proposition}\label{Compl-linDET}
A deterministic linear automaton is complete if and only if for each $q$ of its states, either
\begin{itemize}
\item
  $\delta(q,a,\lambda) = \emptyset$  for all $a\in T$;
 and $\delta(q,\lambda,a) \in Q$  for each $a\in T$; \\
or
\item
  $\delta(q,\lambda,a) = \emptyset$  for all $a\in T$;
and $\delta(q,a,\lambda) \in Q$  for each $a\in T$.
\end{itemize}
\end{proposition}

Moreover, for each deterministic linear automaton, we may construct a complete deterministic linear automaton accepting the same language by adding a sink state, if necessary. This technique is similar to the one used in the case of deterministic finite automata for the regular languages.

Based on the previous proposition and fact we may always assume that our linear automaton accepting a language in 2detLIN is complete.

About the work of linear deterministic automata we state %
the following useful property. It is a kind of analogous property of the deterministic finite automata that when it does the computation on an input $w$, then the initial part of the computation is the same as the computation on a prefix of $w$. %
As linear automata may consume the input from its both extremes, we have a somewhat more complex statement and therefore we state %
it formally. 

\begin{lemma}\label{lemIMPORTANT}
Let a complete deterministic linear automaton $A$ %
 and an input word $w\in T^*$ be given. Let the computation on $w$ by $A$  be $(q_0, w)\Rightarrow^* (q, \lambda)$
such that, the prefix $u$ and the suffix $v$ of $w$ ($w=uv$) were read by the first and second head, respectively, during the computation.
This computation is a $|w|$-step long computation.
Then for any input $u w' v$ with $w'\in T^*$, the first $|w|$ steps of the computation %
are  $(q_0, uw'v) \Rightarrow^* (q, w')$.
\end{lemma}

We note here that in \cite{RAIRO-Zita} for similar models,
 specific functions were defined and used %
 that give the following information for every input:
 which of the heads is stepping in which step of the computation and which head reads the given letter of the input.

Finally, in this section we recall a very important and useful characterization of the regular languages.

Let a language $L \subseteq T^*$ be given. Based on it, we define the equivalence relation: for any $x,y \in T^*$,

 \begin{center}
 $x \equiv_L y \textnormal{  if and only if }
xw \in L \Leftrightarrow  yw\in L \textnormal{ for every } w\in T^*.$
\end{center}
 That is, two words are equivalent if exactly the same continuations of them are in $L$.
The number of the equivalence classes  of the relation $\equiv_L$ is called the index of the language $L$.
By the Myhill-Nerode theorem, a language $L$ is regular if and only if the relation $\equiv_L$  %
has a finite index, i.e., the number of the equivalence classes is finite. Moreover, the index of a regular language
 is then the same as the minimal number of the states in a completely defined finite automaton accepting the language $L$.

In this paper, our aim is to give a kind of similar if and only if characterization of the languages in the class 2detLIN.

\section{Equivalent classes by pairs of prefixes and suffixes}
\label{s:results}

As the computation on %
the input by linear automata goes by reading not only the prefix, but maybe also the suffix of the input word,
we use \textit{pre}fix-\textit{su}ffix-pairs (shortly, \textit{presu}s) in our characterization. Let us consider a language $L$ over the alphabet $T$.
We say that the prefix-suffix-pair $(u_1,v_1)$ is equivalent to the presu $(u_2,v_2)$ with respect to the language $L$,
if for every  word $w\in T^*$, $u_1 w v_1 \in L   \Leftrightarrow     u_2 w v_2 \in L $.
We call a set of equivalence classes of presus a \textit{border classification}, BC for short.
However, a BC not need to cover all prefix-suffix pairs.
We also define \textit{pseudo BC}s, in which in each class, the presus are equivalent to each other, but it may happen that some of the classes contain presus that are also equivalent to each other. From a pseudo BC, a BC can be obtained by joining those classes that contain presus that are equivalent to each other. We say that a (pseudo) BC contains a presu $(u,v)$, if it %
appears in a class of the (pseudo) BC.

To characterize the languages of 2detLIN, we need some additional conditions. In the sequel, we list them.

\begin{definition}
We say that a BC (or a pseudo BC) is \emph{complete}, if for each word $w\in T^*$ it contains exactly one pair $(u,v)$ such that $w = uv$.
\end{definition}

\begin{definition}
We say that a BC (or a pseudo BC) has a \emph{crossing pair}, if it contains both presus $(u_1,v_1)$,  $(u_2,v_2)$
where $u_1$ is a proper prefix of $u_2$ and $v_2$ is a proper suffix of $v_1$.
The presus $(u_1,v_1)$,  $(u_2,v_2)$ are referred as a crossing pair.
\end{definition}

For better understanding these concepts we show some examples.
\begin{example}
Let us consider the regular language $a^*b^*$. One may consider BC $\Omega_1$ with only one class containing all pairs of the form $(a^*,\lambda)$. It is easy to see that $\Omega_1$ is not complete, since, for instance, there is no presu $(u,v)$ in it with $uv = aaab$. On the other hand, $\Omega_1$ does not contain any crossing pairs.

Consider now, the BC $\Omega_2$ with three classes $C_1 = \{(u,v)~|~u\in a^*, v\in b^*\}$ and $C_2 = \{(u,v)~|~ u\in a^*b^*b, v\in b^*\}$ and $C_3 = \{(u,v)~|~ u\in a^*, v\in a a^* b^*\}$.
The BC $\Omega_2$ is not complete, since, e.g., for the word $ab$ it contains the presus $(a,b)\in C_1$ and $(ab,\lambda)\in C_2$. Furthermore, $\Omega_2$ contains crossing pairs, as $(aaa,b) \in C_1$ and $(a,abb) \in C_3$ appear in it.
\end{example}

\begin{definition} Let us fix a language $L$ and a BC for $L$.
The index of the BC is the number of equivalence classes in it.
\end{definition}

The following statement is a direct consequence of the definitions.
\begin{lemma}\label{lem:pseudoBC}
Let a pseudo BC $\Omega$ be given for a language $L$. Then, there is a BC for $L$ that
has index at most the number of classes in the pseudo BC $\Omega$. \\
Further, in general, if for a language $L$ there is a pseudo BC with finitely many classes, then
there is a BC for $L$ with a finite index.
\end{lemma}

We need the following technical lemma that describes an important behaviour of our automata.

\begin{lemma}\label{DETclaim}
By any complete %
 deterministic
   linear automaton $A$, every word $w$ is processed in a unique way and thus, there is exactly one presu $(u,v)$ with $w = uv$ such that $A$ reads $u$ by the first head and $v$ by the second head when %
   performing the computation on the input $w$.
\end{lemma}

Now, we are ready to state and prove one of our main results, the
characterization of 2detLIN languages  by finitely many %
equivalence classes of presus.

\begin{theorem}\label{thm:1}
A language $L$ is in 2detLIN if and only if
 there is a complete BC with a finite index for $L$ that does not contain any crossing pairs.
\end{theorem}
\begin{proof}
 The proof is constructive in both directions.
 First, let us prove that for each language $L$ in 2detLIN, there is a complete BC with finite index as it is stated.

Let $A = (Q,T,q_0,\delta,F)$ be a completely defined deterministic linear automaton accepting $L$ with the set of states $Q=\{q_0,\dots,q_n\}$.
 Based on $A$, we construct a complete pseudo BC. Basically, the construction follows
   Algorithm 1 that is described below. 

\smallskip

\noindent \textbf{Algorithm 1.} \\
\textbf{Input}:  $A = (\{q_0,\dots,q_n\},T,q_0,\delta,F)$, a complete deterministic linear automaton. %
 \\
\textbf{Output}: a pseudo BC for the language accepted by $A$. \\
 Put $(\lambda,\lambda)$ representing the empty word into class $C_0$. \\
 Let the set $J$ of states initially contain only $q_0$ and let the set $J'$ be empty. \\
 While (True) do

 \ \ \ For each $i$  in $\{0, \dots ,n\}$ do

 \qquad   If ($q_i\in J$)

 \qquad \quad   For each $a\in \Sigma$ do

 \qquad \qquad     If ($\delta(q_i,a,\lambda)=q_j$)

 \qquad \qquad \quad   Put $q_j$ into %
 the set $J'$

  \qquad \qquad \quad  For each $(u,v)\in C_i$ do

  \qquad \qquad \qquad  \ \    Put  the presu $((ua,v))$ into $C_j$

 \qquad \qquad     If ($\delta(q_i,\lambda,a)=q_j$)

 \qquad \qquad \quad   Put $q_j$ into %
 the set $J'$

 \qquad \qquad \quad   For each $(u,v)\in C_i$ do

 \qquad \qquad \qquad  \ \          Put the presu $((u,av))$ into $C_j$

  \ \ \ Let $J = J'$ and $J'$ be empty. %

\medskip

 Note that as we have infinitely many presus, the algorithm is running for the infinity, however,
  it puts
   the presus in the appropriate classes by their increasing values of the sum of the lengths of prefix and suffix in a pair. %
    The algorithm works in a somewhat similar manner as a breadth-first search algorithm build an infinite tree level by level.
 Thus,
 for each presu it will be clear after a finitely many steps %
 where it belongs if it appears in the pseudo BC (as we claim it later).

 Clearly the set $J$ contains always a subset of $Q$. %
 It is clear that in the beginning this subset contains only $q_0$. In each iteration of the while loop the new presus appear in the classes that have their sum of the length of prefix and suffix that is one more as similar values of the presus in the previous iteration.
For us, at this moment, the only important is that we can decide which pair appears in the constructed pseudo BC. Moreover, if it appears in it, then we can also decide in which class it is.
See, Claim 1. %

\smallskip

\noindent \textbf{Claim 1.}  The classes obtained by Algorithm~1 form a pseudo BC for the language $L$ accepted by the given complete deterministic linear automaton $A$.

Further, for each presu $(u,v)$ it is clear if it appears in the created pseudo BC, and if so, it is clear  where, into which class $C_j$ it belongs.
Moreover,  the induced pseudo BC is complete. %

\smallskip

\noindent By continuing the proof of the theorem,  it is already clear that the induced pseudo BC is complete.
What is left to be shown is that this pseudo BC does not contain any crossing pairs.

\smallskip

\noindent \textbf{Claim 2.}
For any complete %
deterministic
  linear automaton $A$, the obtained pseudo BC does not contain any crossing pairs.

\smallskip

By the construction, as we have seen, we obtained a pseudo BC,
 if two presus are in the same class $C_i$ then they must be equivalent.
We have proven that there are finitely many classes $C_i$; further the contained presus imply a complete pseudo BC without crossing pairs.
 This, by Lemma \ref{lem:pseudoBC} also proves that there is a complete BC with finite index without crossing pairs, since by joining some classes of the pseudo BC, its completeness and crossing-freeness properties are not changing.
 Thus, the first part of the proof has been finished.

 (We note here that in a pseudo BC some of the sets $C_i$ may contain presus that are equivalent to each other. This property is somewhat similar that a deterministic finite automaton that is not minimal has some states that represent prefix words belonging to the same Myhill-Nerode class.)

\medskip

Now, we prove the other direction. Thus, let us assume that for a language $L$, a complete BC $\Omega$ is given without crossing pairs, then
we define a deterministic linear automaton that accepts $L$ (matching with $\Omega$), %
 and thus the language %
 that is %
  characterized by $\Omega$ %
  is %
  a 2detLIN language. %
Thus, let finitely many equivalence classes $C_1, \dots, C_n $ of presus be given in $\Omega$, %
our aim is to construct a deterministic linear automaton $A = (Q,T,q_0,\delta, F)$ based on that. As the BC is defined for a language, the alphabet $T$ is fixed, %
and it will be used for $A$.
Further, we assign two states $q_i$ and $p_i$ for each class $C_i$. As the given BC is complete, it contains a presu representing the empty word, and it must be $(\lambda,\lambda)$. Let the initial state be one of the states that represents the class $C_i$ which contains $(\lambda,\lambda)$. However, to know which of those, first, we need some technical arguments.

Since no crossing pairs occur in the BC $\Omega$ and it is complete, we can deduce the following statements.

\smallskip

\noindent\textbf{Claim 3.} Let a complete BC for a language be given without crossing pairs.
If $(u,v) \in C_i$ corresponds to the word $uv$ in the BC, then for each $a\in T$ either $(ua,v)$ or $(u,av)$ corresponds to $u a v$. %

\smallskip
When $(u,v) \in C_i$ corresponds to the word $uv$ in a complete BC, then we may also say that $(u,v)$ represents the word $uv$.

\smallskip

\noindent\textbf{Claim 4.}
Let a complete BC for a language be given without crossing pairs.
If $(u,v)$ is in the BC, then either $(ua,v)$ is in the BC for all $a\in T$ or   $(u,av)$ is in the BC for all $a\in T$. %

\smallskip

Based on Claims
 3 and 4, we can now continue our construction.
Applying Claim 4, for the pair $(\lambda,\lambda) \in C_i$, either $(a,\lambda)$ is in $\Omega$ %
 or $(\lambda,a)$.
In the former case, let $q_i$ be the initial state; in the latter case, let $p_i$ be the initial state.
(This is independent of which element $a\in T$ is considered).

Generally, the equivalence classes of the BC $\Omega$ are partitioned into two sets, one containing all presus $(u, v)$ such that $(ua, v)$ is in some $C_j$ in $\Omega$, while the other one contains
the presu $(u, v)$ for $(u, av)$ is in some $C_j$. We label the first mentioned set by the state $q_i$
for the equivalence class $C_i$ and the second one by the state $p_i$. We define
the transition function of the automaton $A$ such that $A$ ends up in state $q_i$
or $p_i$, respectively, if it reads a word $uv$ for which $(u, v)$ is in the according
equivalence class.
 Thus,
for each presu $(u,v)$ and for each letter $a\in T$, let us consider the word $u a v$. As the BC is complete, it appears in the BC represented by exactly one presu, and either the left or the right head reads the last letter $a$ between $u$ and $v$, i.e.,
either $(ua,v)$ or $(u,av)$ appears in the BC, respectively.
 However, it may happen that there are two equivalent  presus $(u_1,v_1)$ and $(u_2,v_2)$ in a class $C_j$ such that for $(u_1,v_1)$ the first, but for
 $(u_2,v_2)$ the second head will make the next read (we show such example later). Therefore, automaton $A$ will be in state $q_j$ after processing words represented by the first type presus, and in $p_j$ after processing words represented by the second type presus.

Formally,
  for each state $q_i$ and for each letter $a\in T$, we define either
\begin{itemize}
\item
 the transition $\delta(q_i,a,\lambda) = q_j $ if there is a presu $(u,v) \in C_i$ such that
 $(ua,v) \in C_j$ and $(uaa,v)$ appears in the BC; or
\item
  $\delta(q_i,a,\lambda) = p_j $ if there is a presu $(u,v) \in C_i$ such that $(ua,v) \in C_j$ and $(ua,av)$
 appears in the BC.
  \end{itemize}

Further,   for each state $p_i$ and for each letter $a\in T$, we define either
 \begin{itemize}
  \item
the transition  $\delta(p_i,\lambda,a) = q_j $ if $(u,v) \in C_i$ and $(u,av) \in C_j$ and  $(ua,av)$
appears in $ \Omega$; or %
\item
 the transition  $\delta(p_i,\lambda,a) = p_j $ if $(u,v) \in C_i$ and $(u,av) \in C_j$ and $(u,aav)$
appears in $\Omega$. %
\end{itemize}
 Clearly,  for each state and letter,
 exactly one of the above transitions will be defined for $A$ based on %
 the properties shown in the previous Claims.

Thus, we can deduce that the transition function determines a complete deterministic linear automaton.

Only one thing is left to define: the set of accepting states $F$. This is based, actually, not on the BC itself, but on some property used to define $\Omega$. In %
$\Omega$, the equivalence classes are defined based on how the possible middle part of the input (i.e., the part we put between the prefix and suffix of the presu) behaves, i.e., with which middle part the input will belong to the language. Now, let
$F =  \{q_i,p_i~|~ C_i \textnormal{  contains  presus  } (u,v)  \textnormal{ such that } uv \in L \}$.

Based on the construction, it can be seen that $A$ accepts the language $L$.
\qed \end{proof}

By the first half of the proof, we are sure that the number of classes in a BC for a 2detLIN language $L$ %
 is not more than %
 the number of states of a complete deterministic linear automaton that accepts $L$.
However, we have seen (by the other direction of the proof) that there could be a BC such that it may require a larger (at most twice much) number of states in an accepting linear automaton. %

We show some examples.
Our first example is very characteristic: %
 the languages of palindromes are in 2detLIN (for any alphabet), but not deterministic linear as
 for alphabets which are at least binary,
 there is no deterministic one-turn pushdown automata accepting them. %
  In fact these languages are 1-rated, i.e., even linear.
\begin{example}
Let us consider the alphabet $T=\{a,b,c\}$. The table of an automaton $A$ that accepts the language of palindromes (the language containing a word $w$ if and only if its reversal $w^R$ is the same as itself) over $T$ is given below in a form of a Cayley table: \\
\begin{center}
\begin{tabular}{l|lllll}
  \hline
  $T \setminus Q$ & $q_0$ (left)& $q_1$ (right) & $q_2$ (right)& $q_3$ (right)& $q_4$ (left)\\ \hline
  $a$ & $q_1$ & $q_0$ & $q_4$ & $q_4$ & $q_4$ \\
  $b$ & $q_2$ & $q_4$ & $q_0$ & $q_4$ & $q_4$ \\
  $c$ & $q_3$ & $q_4$ & $q_4$ & $q_0$ & $q_4$ \\
  \hline
\end{tabular}
\end{center}
Further $q_0$ is the initial state, and $q_0,q_1,q_2,q_3$ are the accepting states, while $q_4$ is, in fact, the sink state. After the name of each state, it is indicated which of the heads can move in transitions from that state.

The equivalence classes of presus based on this automaton are:
\begin{itemize}
\item
 $C_0$: $\{(u,u^R)~|~u\in T^*\}$, \ \ \quad
\item
$C_1$: $\{(ua,u^R)~|~u\in T^*\}$, \qquad
\item
$C_2$: $\{(ub,u^R)~|~u\in T^*\}$,
\item
$C_3$: $\{(uc,u^R)~|~u\in T^*\}$, \quad
\item
$C_4$: $\{(uev,fu^R)~|~u,v\in T^*, e,f\in T, e\ne f\}$.
\end{itemize}
Clearly $(\lambda,\lambda) \in C_0$ and, by applying Algorithm 1, the set $J$ contains only $C_0$. Since the first head can read in $q_0$, the pairs $(a,\lambda)$, $(b,\lambda)$ and $(c,\lambda)$ are created and they are put to classes $C_1,C_2$ and $C_3$, respectively. Then, the new set $J$ contain $C_1,C_2$ and $C_3$. In the next iteration, taking $C_1$ first, some new presus appear in the BC:
$(a,a) \in C_0$, $(a,b), (a,c)\in C_4$ and both $C_0$ and $C_4$ are appended to $J'$. Then $C_2$ and $C_3$ are considered in a similar manner to put some new presus into some classes. Then, updating the set $J$ of states, a new iteration comes.
It can be seen that following the algorithm, the above classes are obtained.
\end{example}

Our next example, is a non-regular, fix-rated linear language over the binary alphabet which can be accepted both by deterministic linear automata and deterministic one-turn pushdown automata.
\begin{example} Now, let us consider the language
$L = \{1^n 0^{3n}~|~n\in\mathbb N\}$, this language is in fact a 3-rated linear language and it is both in detLIN and 2detLIN.

Let us consider the following BC for $L$ (on the left).
\begin{itemize}
\item $C_1$: $\{(1^n,0^{3n})~|~n\in \mathbb N\}$,
\item $C_2$: $\{(1^{n+1},0^{3n})~|~n\in \mathbb N\}$,
\item $C_3$: $\{(1^{n+1},0^{3n+1})~|~n\in \mathbb N\}$,
\item $C_4$: $\{(1^{n+1},0^{3n+2})~|~n\in \mathbb N\}$,
\item $C_5$: $\{(1^{n},0^{3n+1})~|~n\in \mathbb N\}$,
\item $C_6$: $\{(0 w, \lambda) ~|~ w\in \{0,1\}^*\} \cup %
\{(1 w, 1)~|~ w\in \{0,1\}^*\} \cup
\{(1^{m+1} w, 1 0^{3m+1})~|~ m\in \mathbb N, w\in \{0,1\}^*\} \cup
\{(1^{m+1} w, 1 0^{3m+2})~|~ m\in \mathbb N, w\in \{0,1\}^*\} \cup$
$\{(1^{m} w, 1 0^{3m})~|~ m\in \mathbb N, m>0, w\in \{0,1\}^*\} \cup \\ $
$\{(1^{m} 0 w,  0^{3m+1})~|~ m\in \mathbb N, m>0, w\in \{0,1\}^*\}$.
\end{itemize}
It is easy to see that the sets $C_i$ are pairwise disjoint, moreover, the BC is complete as it contains a pair for every word (actually, $C_6$ guarantees this fact).
Based on that we may have the %
 complete deterministic linear automaton $A$ accepting $L$ (see the table below). 
\begin{center}
\begin{tabular}{l|llllll|llllll|}
  \hline
$T \backslash Q$ &$q_1$ &$q_2$ & $q_3$&$q_4$&$q_5$ &$q_6$ &$p_1$ &$p_2$ & $p_3$&$p_4$ &$p_5$ &$p_6$\\ \hline
$0$    & $q_6$ & $-$  & $-$  &$-$&  $q_6$ &$q_6$ &$q_5$ & $p_3$&$p_4$ &$p_1$ &$-$  & $-$ \\
$1$    & $p_2$ & $-$  & $-$  &$-$  &$p_3$ &$q_6$ &$q_6$ & $q_6$&$q_6$ &$q_6$ &$-$  & $-$ \\
  \hline
\end{tabular}
\end{center}
For each state $q_i$ the first, for each state $p_i$ the second head can read the input in the next step.
Further, as $(\lambda,\lambda)\in C_1$ and, e.g., $(1,\lambda)$ in the BC, $q_1$ is the initial state. \\
The final states are $q_1, p_1$ as only class $C_1$ contains presus representing words of $L$.
 Observe that, in fact, the states $q_2,q_3,q_4,p_5,p_6$ are not reachable from $q_0$, thus one may simply erase them from the automaton. %
  Thus, in fact the obtained linear automaton has 7 states (it is complete and deterministic). Observe that class $C_6$ contains the presus that cannot be continued by inserting a word to the middle to get a word of language $L$. Some of the words belonging to these presus are clearly representing something outside of the language, as for instance every word starting with a $0$ is in $\{(0 w, \lambda) ~|~ w\in \{0,1\}^*\}$, or every word ending with a $1$ is either in the above set or in $\{(1 w, 1)~|~ w\in \{0,1\}^*\} $. On the other hand, the presu $(110,0000000)$ is in the set $\{(1^{m} 0 w,  0^{3m+1})~|~ m\in \mathbb N, m>0, w\in \{0,1\}^*\}$, thus it also belongs to $C_6$ even if it represents the word $11 0^8 $, however, ``it was read not in a correct way'' by the heads, thus no continuation of the computation reading it will be accepting.
\end{example}
Neither the automaton nor the characterization by BC, in the previous example, are the simplest one for $L$, however, our aim is to show that our theory works also if not the most efficient description is given if it meets the requirements (e.g., finiteness, completeness). Actually, in the example there are both types of presus in class $C_1$, thus both the states $q_1$ and $p_1$ are required to be in the automaton.

In the next example we highlight the property that a complete deterministic linear automaton may have states for the same class of presus with different head movements.
\begin{example}
Let the language $L$ of the even-length palindromes over $\{a,b\}$ be considered. The following automaton accepts it:
\begin{center}
\begin{tabular}{l|llllllllll}
  \hline
  $T \setminus Q$ & $q_1$ & $p_1$ & $p_2$ & $q_3$& $p_3$ & $p_4$ & $p_5$ & $q_6$ & $q_7$ & $q_8$\\ \hline
  $a$ &             $p_2$ & $q_8$ & $q_1$ & $p_4$ & $q_7$ & $q_7$ & $q_7$ & $q_7$ & $q_7$ & $p_1$ \\
  $b$ &             $q_3$ & $q_6$ & $q_7$ & $p_5$ & $p_1$ & $p_2$ & $p_3$ & $p_1$ & $q_7$ & $q_7$ \\
  \hline
\end{tabular}
\end{center}
where the initial state is $q_1$ and the accepting states are $q_1, p_1$ and $p_5$. For each state $q_i$ the first, for each state $p_i$ the second head can read a letter from the input.

The corresponding classes of presus are belonging to the following languages, i.e., for each class $C_i$, any of the words of $L_i$ can be put into the middle to have a word in $L$.
\begin{enumerate}
\item $C_1$ for states $q_1$ and $p_1$: $L_1 = \{w~|~w \text{ is an even-length palindrome}\}=L$.
\item $C_2$ for state $p_2$: \qquad \quad \ $L_2 = L \cdot\{a\}$.
\item $C_3$ for states $q_3$ and $p_3$: $L_3 = L \cdot \{b\}$.
\item $C_4$ for state $p_4$: \qquad \quad \ $L_4 = L \cdot \{ab\}$.
\item $C_5$ for state $p_5$: \qquad \quad \ $L_5 = L\cdot\{bb\} \cup \{\lambda\}$.
\item $C_6$ for state $q_6$: \qquad \quad \ $L_6 = \{b\} \cdot L$.
\item $C_7$ for state $q_7$: \qquad \quad \ $L_7 = \{\}$, there is no way to make it acceptable.
\item $C_8$ for state $q_8$: \qquad \quad \ $L_8 = \{a\} \cdot L$.
 \end{enumerate}
\end{example}

Finally, we may also use our result to show that a language is not in 2detLIN as we present in the next example.

\begin{example}
Let us consider the language  $L=\{a^nb^nc^n~|~n\in \mathbb N\}$.
We show that $L$ is not a 2detLIN language by contradiction. Thus, let us assume that we have a complete BC without crossing pairs with a finite index for $L$.
 Let the number of equivalence classes be $i$. Further, let us assume that there is a deterministic linear automaton $A$ that accepts $L$ (based on the BC given above).

  There are words in the language with arbitrarily long prefix from $a^*$ and arbitrarily long suffix from $c^*$. Thus, let us consider presus
in the form $(a^m,c^k)$. We show that not any two  different presus in this form can be in the same class. Let   $(a^m,c^k)$ and $(a^j,c^\ell)$ two different presus.
Let us use the notation $max_1=\max\{m,k\}$ %
and let $w = a^{max_1 - m} b^{max_1} c^{max_1-k}$.
Since the two presus are not the same at least one of $m\ne j$ and $k\ne \ell$ holds. Then,
\begin{itemize}
\item on the one hand, presu $(a^m,c^k)$ with the word $w$ results \\ $a^m  a^{max_1 - m} b^{max_1} c^{max_1-k} c^k = 
    a^{max_1} b^{max_1} c^{max_1} \in L$, but
\item on the other hand,  presu $(a^j,c^\ell)$ with the word $w$ results \\  $a^j  a^{max_1 - m} b^{max_1} c^{max_1-k} c^\ell = %
    a^{max_1+j-m} b^{max_1} c^{max_1+\ell -k} $.
   However, in either case, this word is not in $L$, since in the first case, the number of $a$-s does not match with the number of $b$-s, and in the second case, the number of $c$-s does not match with the number of $b$-s.
    \end{itemize}

Considering %
 the word $a^{2i} b^{2i} c^{2i} \in L$, the complete BC must contain at least $2i$ presus of the form $(a^m,c^k)$ that belong to the first $2i$ steps of an accepting computation of this word by the deterministic linear automaton $A$. However, each of these presus must be in a unique class which contradicts to the fact that there are only $i$ classes.
\end{example}

\section{Discussion}
Now, let us discuss what can we gain and what we cannot gain by such characterizations.
For the regular languages, the Myhill-Nerode characterization is closely related to the minimal deterministic finite automaton accepting the language, as we have recalled.
Moreover, as this minimal automaton is unique (up to renaming the states), it also allows to identify a language.

The case of 2detLIN is different, we may have various orders/ways to consume the prefix and the suffix of the input.
 However, we have some strong analogies.
As for the original Myhill-Nerode theorem, an automaton accepting the considered language $L$ may have computations that equivalent words lead the automaton to the same state. Based on the (second half of) the proof of Theorem \ref{thm:1}, we state the following analogous result for 2detLIN languages in the form of a theorem. %

\begin{theorem}\label{thm:MN2}
The BC characterization of a 2detLIN language $L$ allows us to have a deterministic linear automaton $A$ accepting $L$ such that there are at most two states for each equivalent set of presus.  Moreover,
in the computations of any two equivalent presus, after processing these prefix and suffix pairs, $A$ is in one of these two states 
 (let us denote them by $q_i$ and $p_i$ for class $C_i$). If $A$ has both of them, %
  then in one of them the first, in the other the second head can move.
 If a presu $(u_1,v_1)$ is in the %
  class $C_i$, then  each input having the prefix-suffix pair $u_1,v_1$ is processed by $A$ through 
   one of the states $q_i$ or $p_i$: either
  $(q_0, u_1 w v_1) \Rightarrow^* (q_i,w)$ for all $w\in T^*$ or $(q_0, u_1 w v_1) \Rightarrow^* (p_i,w)$ for all $w\in T^*$. If $A$ has both $p_i$ and $q_i$, %
   then there is also a presu $(u_2,v_2)$ in $C_i$, such that
$(q_0, u_2 w v_2) \Rightarrow^* (r,w)$ for all $w\in T^*$, where $r \in \{q_i,p_i\}$, but it differs from the state used for presu
$(u_1,v_1)$.
\end{theorem}

Let us discuss, now, cases where we may have a similarly powerful characterization as the original Myhill-Nerode result for the regular languages. It is proven in \cite{Nagy2013} that all $k$-rated
linear languages for all nonnegative rational values of $k$ are in 2detLIN. More precisely, it is shown that the set of fix-rated linear languages is a proper subset of 2detLIN.

\begin{theorem}\label{k-ratedTHM}
Let us consider a $k$-rated linear language $L$ with $k=\frac{m}{n}$ %
with co-primes $m$ and $n$. Then $L$ has a complete (pseudo) BC without crossing pairs such that
for all presus in the class ``always the same head is stepping'' in a corresponding automaton. More precisely,
if $(u_1,v_1)$ and $(u_2,v_2)$ are both in the class $C_i$, then either both $(u_1a,v_1)$ and $(u_2a,v_2)$ are in the BC, and they are in the same class $C_a$ for each $a\in T$, respectively; or
both $(u_1,av_1)$ and $(u_2,av_2)$ are in the BC, and they are in the same class $C'_a$ for each $a\in T$, respectively.
Moreover, the corresponding complete deterministic linear automaton reads every input with an alternating usage of the heads
 as follows: \\
Till the whole input is processed, %
\begin{itemize}
\item %
 it reads a letter by the first head from the left of the input in $n$ computation steps, then
\item
 in the next $m$ computation steps, it reads the input by the second head from the right.
\end{itemize}
When the last letter is read by a head (depending on the length of the original input), the computation finishes and the acceptance is decided.
\end{theorem}

We \textbf{conjecture} that the minimal automaton (with the parameter $k$) can be defined and determined such
that it has the minimal number of states among the complete deterministic linear automata accepting $L$ and having the above fixed property about the order of the head steps.  Further, this minimal automaton can be used as a unique representant of the given $k$-rated linear language, and thus, also
language equality of these languages can be decided in these classes similarly, as by the original Myhill-Nerode theorem language equivalence of regular languages can be decided.

It is important to use co-primes $m$ and $n$, otherwise the characterization gives a larger number of classes and states. Moreover, the characterization depends on the value of $k$. %
As every regular language is %
$k$-rated with any positive rational value of $k$ (see, e.g., \cite{linPUMP,Semenov}),
this result could give also several alternative characterizations for regular languages.

\begin{corollary}\label{corr}
As for a special subclass, for the regular languages as 0-rated linear languages, exactly the original Myhill-Nerode characterization comes as a special case of our main theorem (Theorem \ref{thm:1}) with Theorem \ref{k-ratedTHM}.
\end{corollary}

Now, we discuss further properties of BCs and coin various open problems.

As each regular language is $k$-rated linear for %
 any positive rational $k$, there is already a large ambiguity to describe them based on Theorem \ref{k-ratedTHM} by fixing the value of $k$ in almost arbitrary way. An interesting question could be how we can find a value of $k$ such that the number of classes will be optimal, i.e., maybe less than their number in the original $k=0$ case. Could it also happen that a minimal representation of a regular language is not connected to any specific value of $k$, that is, the representation does not consider the language as a fix-rated linear?

Now, on the other hand, when a general 2detLIN language is considered,
we know that there is a BC for it that has the finite index property. On the other hand,
there could be various complete BCs without crossing pairs with finite indices for the same language. Thus,
neither the classes, %
 nor their number, nor the number of states of an accepting  complete deterministic linear automaton are uniquely defined.
Therefore, to find the minimal value of classes and/or the minimal number of states of a complete deterministic linear automaton accepting the language are also open questions.

Furthermore, since the linear automata have two heads, we already have some kind of ambiguity based on that, i.e., the order in which the heads process the input may vary from one automaton to other accepting the same language. Moreover, if the order of head movements does not fit for the language, one may also find BC with infinite index representing a 2detLIN language. This can be done, e.g., in the way how a non-regular language is characterized by the original Myhill-Nerode classes:
If one uses in the BC only pairs, where, let us say, the second element, the suffix is always $\lambda$ guessing that the language can be processed by a linear automaton where only the first head is used. We get a complete BC without crossing pairs, but since the language is not regular, this BC has an infinite index (similarly as it has infinite index by using only prefixes).
Therefore, it is crucial to find a kind of efficient representation with a BC to prove that the language is in 2detLIN.

Therefore, we may conclude that in general, we may not be able to identify a 2detLIN language by a given BC.
More precisely, for the same language there are various BCs, but for a BC, the language is
precisely defined if it is also known which of the equivalent classes contain presus $(u,v)$ with the property that $uv\in L$. As we have no bijection between BCs and languages,
  trying to apply this method for language equivalence in general, may need some further techniques to be involved.

Finally, we show  %
  another way how our %
  result is applicable. %
Note that various closure properties of 2detLIN were established in \cite{Nagy2013} and in \cite{UCNC18,ActInf2detLIN}. %

\begin{proposition}\label{prop:complem}
Let $L$ be a 2detLIN language. For $L^c$, the complement of $L$, the same partitions, i.e., equivalence classes can be used as for $L$.
\end{proposition}
\begin{proof}
 Let a completely defined linear automaton $A = (Q,T,q_0,\delta,F)$ for $L$ be given. Then, it has the same set $Q$ of states as a completely defined automaton $A'$ accepting
 $L^c$, with the same transition function. Only the set of accepting states is complemented, i.e., in $A'$ it is $Q\setminus F$. Thus
 based on the transition function and on the set of states, the equivalence classes of presus are the same for these two languages.
\qed  \end{proof}

\section{Conclusions}

The class of sensing $5'\to 3'$ Watson-Crick automata, as well as the class of linear automata, accept
  exactly the linear %
   languages \cite{Rouss,DNA2008,Nagy2013,AFL2017}.
  Their deterministic counterparts accept a class that is a proper superset of the class of regular, but at the same time, it is a proper subset of the class of linear languages. This class is denoted by 2detLIN. Based on deterministic linear automata and on the way they do their computations on the input, we characterized the languages of this class by using equivalent prefix-suffix pairs (abbreviated as presus in the paper, while their partitioning into equivalence classes %
   is abbreviated as BC standing for border classification). We have shown that if there is complete BC for a language $L$ with finitely many equivalence classes without crossing pairs, then the language $L$ is in 2detLIN and vice versa.
  In this way, by our results, on the one hand, the class 2detLIN can be further analysed using this new type of description.
The connection of the number of equivalence classes and the number of states in an accepting minimal complete deterministic linear automaton is not as straightforward as in the case of regular languages. In case of regular languages, the  equivalence classes based only on the prefixes are used and their number is the same as the number of states of a minimal completely defined deterministic finite automaton accepting the language. However,
we believe that the characterization presented here can be connected to a descriptional complexity measure
for 2detLIN, or at least for the class of fixed linear languages, i.e., for a proper superclass  of the set of regular languages. The next steps to this direction are left for future research.
On the other hand, for some languages we are able also to prove that they are not in the class 2detLIN based on our results.

\section*{Acknowledgments} The author is very grateful to the reviewers for their valuable comments.

\bibliographystyle{eptcs}
\bibliography{References}

\end{document}